\documentclass{amsart}

\usepackage[foot]{amsaddr}

\usepackage{subcaption}
\usepackage{tabularx}
\usepackage{amsmath,amssymb,amsthm}
\usepackage[english]{babel}
\usepackage{mathtools}
\usepackage{enumerate}
\usepackage[numbers]{natbib}
\usepackage[hidelinks]{hyperref}
\usepackage{xcolor} 
 
 \usepackage{multirow}

\DeclareMathOperator{\GL}{GL}

\newcommand{\F}{\mathbb{F}}
\newcommand{\Z}{\mathbb{Z}}

\newcommand{\wt}{\mathrm{wt}}
 
\newtheorem{theorem}{Theorem}
\newtheorem{problem}{Problem}

\newtheorem{proposition}[theorem]{Proposition}
\newtheorem{corollary}[theorem]{Corollary}
\theoremstyle{definition}

\newtheorem{definition}[theorem]{Definition}
\theoremstyle{remark}
\newtheorem{remark}[theorem]{Remark}
\newtheorem{example}[theorem]{Example}

\numberwithin{equation}{section}

\newcommand{\zps}{\mathbb{Z} /p^s \mathbb{Z}}
\newcommand{\zpsk}[2][s]{\left(\mathbb{Z}/p^{#1}\mathbb{Z}\right)^{#2}}

\usepackage{stmaryrd}
\newcommand{\gb}{\genfrac{[}{]}{0pt}{}}

\begin{document}

\title{Density of Free Modules over Finite Chain Rings}

\author{Eimear Byrne$^1$}
\address{$^1$ University College of Dublin, Ireland}
\email{ebyrne@ucd.ie}

\author{Anna-Lena Horlemann$^2$}
\address{$^3$University of St. Gallen, Switzerland}
\email{anna-lena.horlemann@unisg.ch}

\author{Karan Khathuria$^3$}
\address{$^4$University of Tartu, Estonia}
\email{karan.khathuria@ut.ee}

\author{Violetta Weger$^1$ $^\star$}
\thanks{$^\star$Corresponding author}
\email{violetta.weger@ucd.ie}







\begin{abstract}
In this paper we focus on modules over a finite chain ring $\mathcal{R}$ of size $q^s$. We compute the density of free modules of $\mathcal{R}^n$, where we separately treat the asymptotics in $n,q$ and $s$. In particular, we focus on two cases: one where we fix the length of the module and one where we fix the rank of the module. In both cases, the density results can be bounded by the Andrews-Gordon identities. We also study the asymptotic behaviour of modules generated by random matrices over $\mathcal{R}$. Since linear codes over $\mathcal{R}$ are submodules of $\mathcal{R}^n$ we get direct implications for coding theory. For example, we show that random codes achieve the Gilbert-Varshamov bound with high probability. 
\end{abstract}

\maketitle

\section{Introduction}
\label{sec:introduction}
The study of the asymptotic behaviour of integer sequences is a classical topic of number theory (see, e.g. \cite{niven}). Such methods rely on obtaining density functions on sets of positive integers in a given interval. More generally, one can study the asymptotic behaviour of combinatorial objects such as linear codes and matrices over a finite alphabet.
There are numerous works in this direction in the case of linear codes and matrices over finite fields \cite{heide,Balakin,partition,fulmangoldstein,gruicarav}, while the study of random code ensembles is as old as the topic of coding theory itself
\cite{bargforney,gallager,loidreau,pierce,shannon}.

In this paper, we consider such questions in the context of linear codes over finite chain rings. A linear code over a finite ring $\mathcal{R}$ is simply an $\mathcal{R}$-module, typically an $\mathcal{R}$-submodule of the free module $\mathcal{R}^n$ where $\mathcal{R}^n$ is endowed with a distance function such as the Hamming, Lee or homogeneous metric.
If $\mathcal{R}$ is not commutative, then the corresponding linear code is usually selected to be a left $\mathcal{R}$-module. 
Codes over rings have been studied widely \cite{assmus, blake1, blake2, hom, Greferath, hammons, klemm, nechaev, shiromoto, spiegel}, but their asymptotic behaviour has not been addressed as thoroughly as their finite field counterparts.  
Formulas on the number of modules over finite chain rings and the matrices that generate them can be found in \cite{butler,count, kschischang,numbermod,  gengauss}. Obtaining estimates on the density of finite modules  of a given shape over a finite chain ring $\mathcal{R}$ can be found by asymptotic enumeration of such formulas in respect of certain parameters, such as the size of the residue field of $\mathcal{R}$, the depth of the ideal chain (also called its nilpotency index)  of $\mathcal{R}$, or the rank of the ambient space. 
Our main results are focused on computing the densities of free modules in $\mathcal{R}^n$; 
in our approach we treat the cases of asymptotic growth of the alphabet $\mathcal{R}$ and the asymptotic growth of $n$ separately. We also separately consider the behaviour of modules of a fixed 
{length} $\ell$, and modules of a fixed {rank} $K$. For the free submodules of $\mathcal{R}^n$ of length $\ell$, we show that as $n$ goes to infinity the density is bounded from above and below by the Andrews-Gordon identities \cite{andrewid,gordon}. As the order $q$ of the residue field of $\mathcal{R}$ goes to infinity, the density of the free codes of length $\ell$ in $\mathcal{R}^n$ goes to $1$. In contrast to this, as the depth $s$ of the ideal chain in $\mathcal{R}$ goes to infinity, the density of the free codes of length $\ell$ in $\mathcal{R}^n$ approaches Euler's function $(1/q)_\infty$. In respect of free modules of rank $K$, we show that the density of such modules as $n$ grows infinitely large is $1$ if $K<n/2$ and is zero if $K>n/2$. 

Related to such problems is the question on the probability that a random matrix generates a code attaining the Gilbert-Varshamov bound. The probabilistic method has a natural application to this topic in coding theory, where density questions can be framed in terms of the behaviour of random matrices.
We show that a random matrix generates a code that achieves the Gilbert-Varshamov bound with probability at least $1-e^{\Omega(n)}$. In this result, we consider a general metric on the ambient space, which is extended to tuples additively. The Hamming, Lee and homogeneous metric are some examples of such a metric.
The interest in this result is two-fold: on one hand it shows the existence of asymptotically good codes, 
and on the other hand it is of importance in the design and analysis of certain code-based cryptosystems based on linear codes over finite chain rings. For the latter, note that in particular Lee-metric codes have recently gained interest in the code-based cryptography community (see for example \cite{bariffi2021analysis,Horlemann2019,leenp}). However, many open questions about these codes have so far obstructed the design and proper cryptanalysis of cryptosystems based on them. One of these questions is the behaviour of codes generated by a random matrix, which we hereby answer.

This paper is organized as follows: in Section \ref{sec:preliminaries} we introduce the definitions and tools that are used throughout this paper. In Section \ref{sec:dens} we give the density results of free modules of given length, respectively of given rank. In Section \ref{sec:codes} we point out the implications of these results for coding theory and show that random codes achieve the Gilbert-Varshamov bound with high probability. Finally, in Section \ref{sec:concl} we state some interesting open problems.

\section{Preliminaries}\label{sec:preliminaries}

In this section we recall some basics of modules over finite chain rings. We also introduce the main tools that are used in this paper, such as the Gaussian  coefficient  and its connection to the $q$-Pochhammer symbol (for more information see for example \cite{bailey, chara, gasper, mac}).

\subsection{Modules over Finite Chain Rings}

We consider finite chain rings and modules over such rings. See, for example, \cite{numbermod,mac} for an introduction to these objects. 
A ring $\mathcal{R}$ is called a left (resp. right) \emph{chain ring} if the left (resp. right) ideals of $\mathcal{R}$
form a chain i.e., for any two left (resp. right) ideals $I, J \subseteq \mathcal{R}$ one has either $I \subseteq J$ or $J \subseteq I$.
Note that in such rings every left ideal is also a right ideal. Every finite chain ring is a local ring, where the maximal ideal is principal. 
Throughout this paper, we let $\mathcal{R}$ denote a finite chain ring and denote by $\langle \pi \rangle$ its unique maximal ideal. We denote by $s$ the nilpotency index of $\mathcal{R}$, i.e., $s$ is the smallest positive integer such that $\pi^s =0$. 
If $q$ is the size of the residue field of $\mathcal{R}$, i.e., if $q= |\mathcal{R}/ \langle \pi \rangle |$ then $|\mathcal{R}| = q^s$.
Well-known examples of finite chain rings include the integer modular ring $\Z/{p^s} \Z$ for $p$ prime and the Galois ring $GR(p^s,r)$ of characteristic $p^s$ and order $p^{rs}$.

Unless explicitly stated otherwise, $M:={_\mathcal{R}}{M}$ will henceforth denote a finite left $\mathcal{R}$-module.

Any module $M$ of our finite chain ring $\mathcal{R}$ can be written as a direct sum of cyclic $\mathcal{R}$-modules, i.e.,
$$M \cong \bigoplus_{i=1}^K \mathcal{R}/ \langle \pi \rangle^{\lambda_i},$$
where $\lambda_1\geq \cdots \geq \lambda_K >0.$ This sequence forms thus a partition $\lambda= (\lambda_1, \ldots, \lambda_K),$ which is called the \emph{type} of a module \cite{mac}. 
The conjugate partition of $\lambda$ is called the \emph{shape} $\mu = (\mu_1,\mu_2,\ldots)$ of the module and is such that 
$$\mu_i = \dim_{\mathcal{R}/ \langle \pi \rangle}(\langle \pi \rangle^{i-1} M / \langle \pi \rangle^i M).$$

The \emph{length} $\ell$ of the module $M$ is equal to the weight of the type $\lambda$, i.e.,  $\ell= \mid \lambda \mid = \sum_{i=1}^K \lambda_i$ and is such that $$\log_q(\mid M \mid) = \ell,$$ whereas the length  $K$ of this partition $\lambda$ is called the \emph{rank} of the module. 

Note that the type $\lambda$ is uniquely determined by the module $M$ and any two modules of type $\lambda$ are isomorphic. Clearly $M$ is a free $\mathcal{R}$-module if it has type $(s, \ldots, s)$ and thus shape $(K, \ldots, K).$


For convenience, we will denote the type $$\lambda =(\underbrace{s, \ldots, s}_{k_1}, \underbrace{s-1, \ldots, s-1}_{k_2}, \ldots, \underbrace{1, \ldots, 1}_{k_s})  $$ of a module, using the frequencies of its parts, i.e.,  $\lambda =(s^{k_1} (s-1)^{k_2} \cdots 1^{k_s}).$

We say that a matrix $A \in \mathcal{R}^{m \times n}$ has type $\lambda$,   if the left $\mathcal{R}$-module generated by its rows has  type $\lambda$.

\subsection{Gaussian Coefficient}

We now introduce some standard formulas that we will use in our counting arguments.

\begin{definition}
The \emph{Gaussian coefficient} or the $q$-binomial coefficient is defined as  
$$\gb{n}{k}_q = \prod\limits_{i=0}^{k-1} \frac{q^n-q^i}{q^k-q^i} = \frac{\prod_{i=1}^n (1-q^i)}{\prod_{i=1}^k (1-q^i)\prod_{i=1}^{n-k} (1-q^i)}.$$
\end{definition}
If $q$ is a prime power, then the Gaussian coefficient gives the number of subspaces of $\mathbb{F}_q^n$ of dimension $k$. 

\begin{definition}
Let $r$ be a positive integer. The \emph{$q$-Pochhammer symbol} is defined to be 
\begin{align*} (a;q)_r &:= \prod_{i=0}^{r-1} \left(1-aq^i\right). 
\end{align*}
We also define $(a;q)_0:=1$ and in the instance that $0 < q <1$ for a fixed $a \in \mathbb{R}$, we define 
\begin{align*}
\lim\limits_{r \to \infty} (a;q)_r =(a;q)_\infty & := \prod_{i=0}^\infty \left(1-aq^i\right).
\end{align*}
\end{definition}
Clearly, $(q;q)_r = \prod_{i=1}^r \left(1-q^i\right)$, and since we will use this quantity often, we will abbreviate $(q;q)_r$ to $(q)_r$.

The Gaussian coefficient has the following expression in terms of $q$-Pochhammer symbols: 
$$\gb{n}{k}_q = \frac{(q)_n}{(q)_{k}(q)_{n-k}}.$$
Regarding the asymptotics of the Gaussian coefficient, we have that for a constant $k\in \{1,\ldots, n\}$ and $| 1/q | <1$
$$\lim\limits_{n  \to \infty} \gb{n}{k}_{1/q} = \frac{1}{(1/q)_k}.$$
For fixed $0 <R<1$  we have that \begin{equation}
    \lim\limits_{n \to \infty} \gb{n}{Rn}_{1/q} = \frac{1}{(1/q)_\infty} , \label{eq:asymp_Gauss}
\end{equation}   
whereas for $R \in \{0,1\}$, this limit is 1, as  $\gb{n}{Rn}_{1/q} =1$.
An important identity of the Gaussian coefficient is the following:
$$ \gb{n}{k}_q = q^{(n-k)k} \gb{n}{k}_{1/q}.$$
It was shown in \cite{gadouleau}, that
$$q^{(n-k)k} \leq \gb{n}{k}_q \leq \frac{1}{(1/q)_\infty} q^{(n-k)k}. $$

The factor $(1/q)_\infty$, also called Euler's function, can be interpreted as the ratio of invertible $m \times m$ matrices to all matrices, i.e., $\frac{\mid \GL_m(\mathbb{F}_q)\mid}{\mid \mathbb{F}_q^{m \times m} \mid}$ for $m$ going to infinity. We have that $(1/q)_\infty$ is increasing in $q$ and approaches $1$ as $q \to \infty$. In particular, $ \gb{n}{k}_q \sim q^{(n-k)k}$ as $q \to \infty$.

\subsection{Counting Modules}\label{sec:countnew}
 
In this subsection we will provide the main formulas for counting the number of modules of a given length or rank over a finite chain ring $\mathcal R$. 

The following proposition is
a special case of \cite[Theorem 2.4]{numbermod}. A proof for the case $\mathcal{R}= \zps$ can be read in \cite{butler}.

\begin{proposition}
The number of submodules of $\mathcal{R}^n$ with shape $\mu=(\mu_1,\dots,\mu_s)$ is given by 
$$N_{n,q}(\mu) := \prod\limits_{i=1}^{s} q^{(n - \mu_{i})\mu_{i+1}} \gb{n-\mu_{i+1}}{\mu_{i}-\mu_{i+1}}_q = q^{\sum\limits_{i=1}^s (n-\mu_i) \mu_{i+1}} \prod\limits_{i=1}^s\gb{n-\mu_{i+1}}{\mu_i-\mu_{i+1}}_q,$$
\end{proposition}
where we define $\mu_{s+1}=0.$
If we consider the type, then the number of submodules of $\mathcal{R}^n$ with type $(s^{k_1} \cdots 1^{k_s})$ is given by 
 $$N_{n,q}(k_1, \ldots, k_s)  :=  q^{\sum_{i=1}^s (n - \sum_{j=1}^i k_j)\sum_{j=1}^{i-1}k_j} \prod\limits_{i=1}^s \gb{n-\sum_{j=1}^{i-1}k_j}{k_i}_q.$$

 The number of free submodules of rank $K$ 
 is then given by 
   $$ N_{n,q}(K, 0, \ldots, 0)= q^{(n-K)K(s-1)} \gb{n}{K}_q.$$

We define $L(s,n,\ell)$ to be the set of all possible types for length $\ell$, i.e.,
$$L(s,n,\ell) := \left\{(k_1, \ldots, k_s) \mid \sum\limits_{i=1}^s k_i (s-i+1) =\ell, \sum\limits_{i=1}^s k_i \leq n \right\}.$$ 
We thus count all submodules  of $\mathcal{R}^n$ of length $\ell$ as
$$M(n,\ell,q,s) := \sum\limits_{(k_1, \ldots, k_s) \in L(s,n,\ell)} N_{n,q}(k_1, \ldots, k_s).$$

Using the connection of the type to the shape $\mu$ of a module $M$, we see that $M(n,\ell,q,s)$ is a $q$-multinomial as defined in \cite[Definition 1]{qmultinomials}: 
\begin{definition}\label{qmulti} Let $\ell \leq n$ be non-negative integers, let $s$ be a positive integer and let $q>0$. The \emph{$q$-multinomial} is defined as
\begin{equation}
  \gb{n}{\ell}_q^{(s)} := \sum\limits_{\mu_1+ \cdots + \mu_s = \ell} q^{\sum_{j=1}^{s-1} (n-\mu_j)\mu_{j+1}} \gb{n}{\mu_1}_q \gb{\mu_1}{\mu_2}_q \cdots \gb{\mu_{s-1}}{\mu_s}_q. \label{eq:qmultinomial}
\end{equation} 
\end{definition}
Note that this definition is different to the usual definition of the $q$-multinomial coefficient (see for example \cite{chara}), which is defined as $$\gb{n}{i_1, \ldots, i_r, n-K}_q = \gb{n}{K_r}_q \gb{K_r}{K_{r-1}}_q \cdots \gb{K_2}{K_1}_q,$$  where $\sum_{j=1}^m i_j =K_m.$

In addition, we need to count the number of modules of a given rank, for which we need the set of weak compositions of $K$ into $s$ parts, denoted by $C(s,K)$, i.e., $$C(s,K):= \left\{ (k_1, \ldots, k_s) \mid 0 \leq k_i \leq K, \sum\limits_{i=1}^s k_i=K\right\}.$$

The number of submodules  of $\mathcal{R}^n$ of rank $K$ is given by
\begin{align*} W(n,K,q,s) := \sum\limits_{(k_1, \ldots, k_s) \in C(s,K)} N_{n,q}(k_1, \ldots, k_s).
\end{align*}

\section{The Density of Free Modules}\label{sec:dens}
In this section we will determine the densities 
of free modules of a finite chain ring $\mathcal R$. In the first subsection we  fix the length of the modules and in the second we  fix the rank of the modules. 
By density, we refer to the limit of a probability as a chosen parameter goes to infinity. More precisely, for any two sequences $(A(n))_{n \geq 1}$ and $(S(n))_{n \geq  1}$ of sets satisfying $A(n) \subseteq S(n)$, we denote by $$\lim_{n \to \infty} \frac{| A(n)|}{|S(n)|},$$ the {\em density} of $A$ in $S$, if the limit exists. Furthermore, if the density of $A$ in $S$ is 0, we say that $A$ is {\em sparse} in $S$ and if the density of $A$ in $S$ is 1, we say that $A$ is {\em dense} in $S$.

\subsection{Density of Free Modules of Fixed Length}\label{sec:densk}

For the remainder, we let $\mathcal{R}$ have a residue field of order $q$. Recall that we denote by $s$ the nilpotency index of $\mathcal{R}$.

The probability that a randomly chosen submodule of $\mathcal{R}^n$ of length $\ell$ is free is given by

$$\psi(n,\ell,q,s) := \frac{ N_{n,q}(\ell/s, 0, \ldots, 0) }{M(n,\ell,q,s) }= \frac{q^{(sn-\ell)\ell(s-1)/s^2} \gb{n}{\ell/s}_q }{M(n,\ell,q,s) }.$$
Note that, for a free module, $\ell$ is a multiple of $s$. 
In order to rewrite this probability, we introduce the following notation, adapted from Equation (2.17)   of \cite{schilling1996multinomials}: 
\begin{equation}
T_s(n,m;q) := q^{sn^2/4-m^2/s} \gb{n}{sn/2-m}_{1/q}^{(s)}. \label{eq:T_multinomial}
\end{equation}

Using the fact that $\gb{n}{k}_q= \gb{n}{k}_{1/q} q^{(n-k)k}$ and \eqref{eq:T_multinomial}, we can rewrite the probability $\psi(n,\ell,q,s)$ as

\begin{align*} 
\psi(n,\ell,q,s) &= \frac{ \gb{n}{\ell/s}_{1/q} }{q^{-(ns-\ell)\ell/s}M(n,\ell,q,s) }  = \frac{\gb{n}{\ell/s}_{1/q}}{T_s(n,sn/2-\ell;1/q)}.
\end{align*}
The quantity $T_s(n, sn/2-\ell; 1/q)$ was shown in \cite[Corollary 2.1]{supernom} to have the following limit 
 \begin{align}\label{limiteq} &\lim\limits_{n \to \infty} T_s(n,sn/2-\ell;1/q)= \nonumber \\
 & \frac{1}{(1/q)_\infty}  \sum\limits_{\substack{(k_2, \ldots, k_s) \in \Z_{\geq 0}^{s-1} \\ 1/s(k_2(s-1) + k_3(s-2) + \cdots + k_s) \in \mathbb{Z}}} \frac{(1/q)^{(k_2,  \dots , k_s) C^{-1} (k_2, \ldots, k_s)^\intercal}}{(1/q)_{k_2} \cdots (1/q)_{k_s}},   \end{align}
 
where $C$ is the $(s-1) \times (s-1)$ Cartan matrix defined via $C_{i,j}^{-1} = \min\{i,j\} - \frac{ij}{s}$.

In the following we study the asymptotics of $\psi(n,\ell,q,s)$, thus finding the density of free modules of a given length. We start with the asymptotics in $n$, where we let $\ell= Rns$ for $0 < R <1$ and $s \mid \ell.$

\begin{theorem}\label{thm:dens}
Let $\ell $ and $n$ be positive integers with $\ell=Rns$, where $0<R<1$ and $ s \mid \ell.$
The density as $n \to \infty$ of the free submodules of $\mathcal{R}^n$ of length $\ell$ is given by
\begin{equation}\label{denskps}
    \left(\sum\limits_{\substack{k_2, \ldots, k_s \geq 0\\ s \mid K_2 + \cdots + K_s }} \frac{(1/q)^{K_2^2 + \cdots+ K_s^2  -(K_2+ \cdots +K_s)^2/s  }}{(1/q)_{k_2} \cdots (1/q)_{k_s}} \right)^{-1},
\end{equation}
where $K_i=\sum_{j=2}^i k_j.$ 
\end{theorem}

\begin{proof}
From \eqref{eq:asymp_Gauss} and
\eqref{limiteq} it follows that 
\begin{align*}\lim\limits_{n \to \infty}\psi(n,\ell,q,s)
&= \lim\limits_{n \to \infty} \frac{\gb{n}{\ell/s}_{1/q}}{T_s(n,sn/2-\ell;1/q)}\\
& =\left( \sum\limits_{\substack{(k_2, \ldots, k_s) \in \Z_{\geq 0}^{s-1} \\ 
\frac{1}{s}(k_2(s-1) + k_3(s-2) + \cdots + k_s) \in \mathbb{Z}}} \frac{(1/q)^{(k_2,  \dots , k_s) C^{-1} (k_2, \ldots, k_s)^\intercal}}{(1/q)_{k_2} \cdots (1/q)_{k_s}}\right)^{-1}, \end{align*}
where $C_{i,j}^{-1} = \min\{i,j\} - \frac{ij}{s}$. 
The condition on the sum is equivalent to 
$$s \mid k_2(s-1)+k_3(s-2)+ \cdots + k_s.$$ 
It can be proved by induction that: 

\begin{align*}
    (k_2,  \dots ,k_s) C^{-1} (k_2, \ldots, k_s)^\intercal  &= \frac{1}{s} \left( \sum_{i=2}^{s} k_i^2 + \sum_{i=2}^{s-1} (k_i +k_{i+1})^2 + \cdots + (k_2 + \cdots +  k_s)^2 \right) \\
    &=\sum_{i=2}^s (i-1) \frac{s-i+1}{s} k_i^2 + \sum\limits_{2 \leq i <j \leq  s} 2k_ik_j (i-1)\frac{s-j+1}{s} \\
    &= \sum_{i=2}^s \left( \sum_{j=2}^i k_j \right)^2 -\left( \sum_{i=2}^s \frac{s-i+1}{s} k_i \right)^2 s \\ 
    & = K_2^2 + \cdots+ K_s^2  - \frac{(K_2+ \cdots +K_s)^2}{s}, 
\end{align*}
where we set $K_i= \sum_{j=2}^i k_j$.
Hence the density of the free modules of length $\ell$ is given by 
\begin{align*}
    &  \left(\sum\limits_{\substack{k_2, \ldots, k_s \geq 0\\ s \mid K_2 + \cdots + K_s }} \frac{(1/q)^{K_2^2 + \cdots+ K_s^2  -(K_2+ \cdots +K_s)^2/s  }}{(1/q)_{k_2} \cdots (1/q)_{k_s}} \right)^{-1}.
\end{align*}
\end{proof}

Note that this is very similar to the Andrews-Gordon identity \cite[Theorem 1]{andrewid}, which states that for $ \mid q \mid <1$

\begin{align*}
    AGI(q,s ):=& \sum\limits_{n_1, \ldots, n_{s-1} \geq 0} \frac{q^{N_1^2+ \cdots N_{s-1}^2}}{(q)_{n_1} \cdots (q)_{n_{s-1}}}
    =  \prod\limits_{n \not\equiv 0, \pm i \mod 2s+1}^\infty \left(1-q^n\right)^{-1} \\
      =  & \frac{(q^s;q^{2s+1})_\infty (q^{s+1}; q^{2s+1})_\infty (q^{2s+1};q^{2s+1})_\infty}{(q)_\infty},
\end{align*}
for any $i \in \{1, \ldots, s\}$ and $N_i = n_i + \cdots +n_{s-1}.$
It may be the case that the series in \eqref{denskps} has a similar expression to the above in terms of $q$-Pochhammer symbols.
Achieving this would require a generalization of the Andrews-Gordon identity, which is itself a generalization of the Roger-Ramanujan identities. 
Another generalization of the Roger-Ramanujan identities that can be considered is the Alder polynomial \cite{alder}. This generalization is essentially the Andrews-Gordon identity along with the condition that $s\mid (N_1+ \cdots + N_{s-1})$, see \cite[Theorem 2]{andrewid}. Since a suitable formulation of the Alder polynomial in terms of $q$-Pochhammer symbols is not known, we only consider the Andrews-Gordon identity for our results.

In the special case where $s=2$, we are indeed able to express this density  in terms of $q$-Pochhammer symbols: in fact, in Equation (2.53) of \cite{baxter} it is shown that for even $n-m$:
$$\lim\limits_{n \to \infty} T_2(n,n-m;1/q) = \frac{1/2 ((-\sqrt{1/q};1/q)_\infty+(\sqrt{1/q};1/q)_\infty)}{(1/q)_\infty}.$$
Since we set $m=n-\ell$, this is indeed the case and thus, the density for $s=2$ is given by 
\begin{align*}
\frac{ 2 }{(-\sqrt{1/q};1/q)_\infty+(\sqrt{1/q};1/q)_\infty }. 
\end{align*}

In the case $s>2$ however, we can give bounds on \eqref{denskps} using the Andrews-Gordon identity.

\begin{theorem}\label{densbound}
Let $\mathcal{R}$ have a residue field of size $q$ and $s$ be its nilpotency index. Let $\ell$ and $n$ be positive integers with $\ell=Rns$, where $0<R<1$ and $s \mid \ell$.
The density as $n \to \infty$ of free submodules in $\mathcal{R}^n$ of length $\ell$ given by \eqref{denskps} can be bounded as follows:
\begin{align*}
    AGI\left(1/q,s\right)^{-1} \leq \left(\sum\limits_{\substack{k_2, \ldots, k_s \geq 0 \\ s \mid K_2 + \cdots + K_s}} \frac{(1/q)^{K_2^2 + \cdots+ K_s^2  -(K_2+ \cdots +K_s)^2/s  }}{(1/q)_{k_2} \cdots (1/q)_{k_s}} \right)^{-1} 
    \leq AGI(1/q',s)^{-1},
\end{align*}
for $q' := q^{s^2-s}.$
\end{theorem}
 
 \begin{proof}
 We first show that 
 \begin{align*}
  \left(\sum\limits_{\substack{k_2, \ldots, k_s \geq 0 \\ s \mid K_2 + \cdots + K_s}} \frac{(1/q)^{K_2^2 + \cdots+ K_s^2  -(K_2+ \cdots +K_s)^2/s  }}{(1/q)_{k_2} \cdots (1/q)_{k_s}} \right)^{-1} 
    \geq AGI(1/q,s)^{-1}. 
\end{align*}
For this, we first note that
$$ K_2^2 + \cdots +K_s^2 - (K_2 + \cdots +K_s)^2/s > \frac{1}{s} \left( K_2^2 + \cdots +K_s^2\right).$$
This is equivalent to 
$$(s-1) (K_2^2+ \cdots + K_s^2)  > (K_2 + \cdots + K_s)^2,$$
which follows from the Cauchy-Schwarz inequality. Note, that there always exists some $a,b \in \mathbb{N}$, such that $\frac{1}{s} > \frac{a^2}{b^2},$ with $b \mid K_i$ for all $i \in \{2, \ldots, s\}$. For example one can always choose $a=1$ and $b=sK_1 \cdots K_s.$
We define $n_i= \frac{a}{b} k_i$ and $N_i= \sum_{j=2}^i n_j,$ which is thus $N_i = \frac{a}{b} K_i$ for all $i \in \{2, \ldots, s\}$.
Since $n_i < k_i$ and hence $(1/q)_{k_i} < (1/q)_{n_i}$, we can bound our sum  using the observation above and the Andrews-Gordon identity as
\begin{align*}
     \sum\limits_{\substack{k_2, \ldots, k_s \geq 0 \\ s \mid K_2 +  \cdots +  K_s}} \frac{(1/q)^{K_2^2 + \cdots+ K_s^2  -(K_2+ \cdots +K_s)^2/s  }}{(1/q)_{k_2} \cdots (1/q)_{k_s}} 
    \leq &  \sum\limits_{\substack{k_2, \ldots, k_s \geq 0 \\ s \mid K_2 +  \cdots + K_s}} \frac{(1/q)^{\frac{a^2}{b^2}(K_2^2 + \cdots+ K_s^2)}}{(1/q)_{k_2} \cdots (1/q)_{k_s}}   \\
        \leq &  \sum\limits_{k_2, \ldots, k_s \geq 0 } \frac{(1/q)^{\frac{a^2}{b^2}(K_2^2 + \cdots+ K_s^2)}}{(1/q)_{k_2} \cdots (1/q)_{k_s}}   \\
       \leq &  \sum\limits_{n_2, \ldots, n_s \geq 0 } \frac{(1/q)^{N_2^2 + \cdots+ N_s^2 }}{(1/q)_{n_2} \cdots (1/q)_{n_s}}  \\
       =& ~ AGI(1/q,s).
\end{align*}

We now show that 
\begin{align*}
     \left(\sum\limits_{\substack{k_2, \ldots, k_s \geq 0 \\ s \mid K_2 + \cdots + K_s}} \frac{(1/q)^{K_2^2 + \cdots+ K_s^2  -(K_2+ \cdots +K_s)^2/s  }}{(1/q)_{k_2} \cdots (1/q)_{k_s}} \right)^{-1} \leq AGI\left(1/q^{s^2-s},s\right)^{-1}.
\end{align*}
In order to give this upper bound on the density, we lower bound the series by taking fewer summands. More precisely, we assume that $s \mid K_i$ for all $i \in \{2, \ldots, s\}.$ Thus, we can introduce new variables $N_i=K_i/s$ as well as $n_i= k_i/s$ for all $i \in \{2, \ldots, s\}.$ 
We can thus bound the density with $AGI(1/q',s)^{-1},$ where $q'= q^{s^2-s}.$
\begin{align*}
     \sum\limits_{\substack{k_2, \ldots, k_s \geq 0 \\ s \mid K_2 + \cdots + K_s }} \frac{(1/q)^{K_2^2 + \cdots+ K_s^2  -(K_2+ \cdots +K_s)^2/s  }}{(1/q)_{k_2} \cdots (1/q)_{k_s}} 
    \geq  &  \sum\limits_{\substack{k_2, \ldots, k_s \geq 0 \\ s \mid K_2, \ldots, s \mid K_s }} \frac{(1/q)^{K_2^2 + \cdots+ K_s^2 -(K_2+ \cdots +K_s)^2/s  }}{(1/q)_{k_2} \cdots (1/q)_{k_s}} \\
    \geq  &  \sum\limits_{k_2, \ldots, k_s \geq 0 } \frac{(1/q)^{(s^2-s)(N_2^2 + \cdots+ N_s^2)  }}{(1/q)_{k_2} \cdots (1/q)_{k_s}} \\
  \geq  &  \sum\limits_{n_2, \ldots, n_s \geq 0 } \frac{(1/q^{s^2-s})^{N_2^2 + \cdots+ N_s^2 }}{(1/q)_{n_2} \cdots (1/q)_{n_s}} \\
    \geq  &  \sum\limits_{n_2, \ldots, n_s \geq 0 } \frac{(1/q')^{N_2^2 + \cdots+ N_s^2 }}{(1/q')_{n_2} \cdots (1/q')_{n_s}} \\ 
    = & ~  AGI(1/q',s),
\end{align*}
where we have used that $(1/q)_{k_i} \leq (1/q)_{n_i} $ and $(1/q)_{n_i} \leq (1/q^{s^2-s})_{n_i}$ for all $i \in \{2, \ldots, s\}.$ 
Hence, we get the claim.
\end{proof}

Values for the upper and lower bound, as well as the exact values can be found in Table \ref{table:densities}. 
This shows that the density for growing $n$ of free modules of a given length over a finite chain ring is neither dense nor sparse, if $q$ and $s$ are fixed.

\begin{table}[ht]
 \begin{center}
 \begin{tabular}{|c|c| c|c|c|}
 \hline 
  &  & Lower Bound &  Exact & Upper Bound  \\\hline
 \multirow{5}{*}{$s=2$} &  $q=2$ & 0.46026 & 0.59546 & 0.74688 \\
  & $q=3$ & 0.65750  & 0.84191 & 0.88752 \\
  & $q=5$ & 0.79867  & 0.95049 & 0.95999 \\
  & $q=7$ & 0.85678 & 0.97627 & 0.97959 \\
  & $q=11$ &  0.90903 & 0.99092 & 0.99173 \\ \hline 
 \multirow{5}{*}{$s=3$} &  $q=2$ & 0.35536  & 0.47084 & 0.98413 \\
  & $q=3$ & 0.58922 & 0.79666 & 0.99862  \\
  & $q=5$ & 0.76770 & 0.94102 & 0.99994\\
  & $q=7$ & 0.83959 & 0.97295 & $1-8.5 \cdot 10^{-6}$\\
  & $q=11$ & 0.90157 & 0.99010 & $1-5.6 \cdot 10^{-7}$\\ \hline 
   \multirow{5}{*}{$s=4$} &  $q=2$ & 0.31866 & 0.42109 & 0.99976 \\
  & $q=3$ & 0.56950 & 0.78230 & $1-1.8 \cdot 10^{-6}$\\
  & $q=5$ & 0.76180 & 0.93915 & $1-4.1 \cdot 10^{-9}$\\
  & $q=7$ & 0.83719 & 0.97248 & $1-7.2 \cdot 10^{-11}$\\
  & $q=11$ & 0.90090 & 0.99023 &  $1-3.2 \cdot 10^{-13}$\\ \hline 
  \end{tabular}
\end{center}  
  \caption{Density of free modules  in $\mathcal{R}^n$ of a given length}
  \label{table:densities}
\end{table}

\begin{remark}
 It can easily be seen that the probability for a submodule  to be not free, is upper bounded by the counter probability of a submodule to be free. Explicitly,  the density of submodules having type $(s^{k_1} \cdots 1^{k_s})$ such that $(k_1, \ldots, k_s) \in L(s,n,\ell)$ with $s  k_1 \neq \ell$ is upper bounded by
$1- AGI(1/q,s)^{-1}.$
\end{remark}

\vspace{0.5cm}

As a next step we consider the behaviour of $\psi(n,\ell,q,s)$ for large $q$.
It turns out, that for $ q \to \infty$ we have that modules are with high probability free. This is quite intuitive, as for increasing $q$, we are closer to a finite field structure. While the opposite is true for increasing $s$; thus we expect that the probability of having a free module decreases when $s$ increases.

  In the following theorem we prove a stronger statement that   implies the above remark. The question we want to answer reads as follows: for a fixed choice of $n,\ell,s$ and an $0<\varepsilon<1$, how large should we choose $q$ (in terms of $n,\ell,s$ and $\varepsilon$) such that the probability of having a free submodule  of $\mathcal{R}^n$ of given length $\ell$   is at least $1- \varepsilon$?

 \begin{theorem}\label{largep}
Let $\mathcal{R}$ have a residue field of size $q$ and nilpotency index $s$. Let $0<\varepsilon <1$. Then for $q$ satisfying $ AGI(1/q,s)^{-1} \geq 1-\varepsilon$, 
we have that the probability of a random submodule of $\mathcal{R}^n$ of length $\ell$ to be free is at least $1- \varepsilon.$
 \end{theorem}

 \begin{proof}
 For $\ell \leq ns$ positive integers with $s \mid \ell$, the probability that a submodule of $\mathcal{R}^n$ is free is given by

  $$ \psi(n,\ell,q,s) =  \frac{\gb{n}{\ell/s}_{1/q}}{ \sum\limits_{(k_1, \ldots, k_s) \in L(s,n,\ell)} q^{\ell^2/s - \sum\limits_{i=1}^s \left(\sum\limits_{j=1}^i k_j\right)^2}  \frac{(1/q)_n}{(1/q)_{n-K}(1/q)_{k_1} \cdots (1/q)_{k_s}}},$$
 where we have used $\gb{n}{k}_q=q^{(n-k)k} \gb{n}{k}_{1/q}$.
 
We can hence write
\begin{align*}
     \psi(n,\ell,q,s) & = \left( \sum\limits_{(k_1, \ldots, k_s) \in L(s,n,\ell)} q^{\ell^2/s - \sum_{i=1}^s \left( \sum_{j=1}^i k_j \right)^2} \frac{(1/q)_{n-\ell/s} (1/q)_\ell}{(1/q)_{n-K} (1/q)_{k_1} \cdots (1/q)_{k_s}}  \right)^{-1}.
      \end{align*}

 Note that for $b \leq a$ we have $(1/q)_a \leq (1/q)_b$, thus we get that 
     $$ \psi(n,\ell,q,s)^{-1} \leq \sum\limits_{(k_1, \ldots, k_s) \in L(s,n,\ell)} q^{\ell^2/s - \sum\limits_{i=1}^s \left(\sum\limits_{j=1}^i k_j\right)^2}\frac{1}{(1/q)_{k_2} \cdots (1/q)_{k_s}}.$$

Since this is the finitization of, and thus smaller than, \eqref{denskps}, we get 
$$\psi(n,\ell,q,s)^{-1} \leq \sum\limits_{\substack{k_2, \ldots, k_s \geq 0 \\ s \mid k_2(s-1) + \cdots + k_s}  } \frac{(1/q)^{K_2^2 + \cdots K_s^2 -(K_2 + \cdots + K_s)^2/s}}{(1/q)_{k_2} \cdots (1/q)_{k_s}}.$$
Using Theorem \ref{densbound} we can bound this probability further using the Andrews-Gordon identity, i.e.,
$$\psi(n,\ell,q,s)\geq AGI(1/q,s)^{-1} \geq 1- \varepsilon.$$
 \end{proof}

 Since $ AGI(1/q,s)^{-1}  \geq (1/q)_\infty$ we get the following statement, which does not depend on $s$.
 
 \begin{corollary}
Let  $\mathcal{R}$ have a residue field of size $q$ and $s$ be its nilpotency index and let $\ell< s n$ be positive integers with $s \mid \ell$. The probability for a submodule of $\mathcal{R}^n$ of length $\ell$ to be free is at least $(1/q)_\infty$.
 \end{corollary}

 Since $(1/q^s;1/q^{2s+1})_\infty, (1/q^{s+1}; 1/q^{2s+1})_\infty$ and also $(1/q^{2s+1})_\infty$ go to 1 as $q$ or $s$ go to infinity, and $(1/q)_\infty$ goes to $1$ for $q\rightarrow \infty$,  also the intuitive remarks  on the limits in $s$ and $q$ now follow as corollaries. 
 \begin{corollary}
Let $\mathcal{R}$ have a residue field of size $q$ and let $s$ be its nilpotency index. 
\begin{enumerate}
    \item 
The density of free submodules  in $\mathcal{R}^n$  of length $\ell$ for $q \to \infty$ is 1.

\item
The density of free submodules in $\mathcal{R}^n$ of length $\ell$ for $s \to \infty$  is at least $(1/q)_\infty$.
\end{enumerate}
 \end{corollary}

\subsection{Density of Free Modules of Given Rank}
As before, let $\mathcal{R}$ have a residue field of size $q$ and $s$ be its nilpotency index. 
In this subsection we study the asymptotics of 
$$\varphi(n,K,q,s) := \frac{q^{(n-K)K(s-1)} \gb{n}{K}_q}{W(n,K,q,s)},$$ 
i.e., the density of free modules of given rank $K$.  We first rewrite $\varphi(n,K,q,s)$, as follows
\begin{align}
    & \varphi(n,K,q,s) \nonumber \\  & = \frac{\frac{(1/q)_n}{(1/q)_{n-K}(1/q)_K}}{ q^{-s(n-K)K} \sum\limits_{(k_1, \ldots, k_s) \in C(s,K)} q^{n \ell- \sum_{i=1}^s\left(\sum_{j=1}^i k_j\right)^2} \frac{(1/q)_n}{(1/q)_{n-K} (1/q)_{k_1} \cdots (1/q)_{k_s} } } \nonumber \\
    & = \left((1/q)_K \sum\limits_{(k_1, \ldots, k_s) \in C(s,K)}q^{n\ell-snK+sK^2} \frac{(1/q)^{\sum_{i=1}^s \left(\sum_{j=1}^i k_j \right)^2}}{(1/q)_{k_1} \cdots (1/q)_{k_s}} \right)^{-1} \label{varphi}.
\end{align}

In contrast to the density from Theorem \ref{thm:dens}, we have that the asymptotics of $\varphi(n,K,q,s)$ depends on the relation between $K$ and $n$. 
  
  \begin{theorem}
  Let $\mathcal{R}$ have a residue field of size $q$ and $s$ be its nilpotency index.  Let $K$ and $n$ be positive integers with $K= R'n$, where $1/2 < R'< 1$. 
 The density of   free submodules of $\mathcal{R}^n$ of given rank $K$ for $n \to \infty$  is $0$.
  \end{theorem}
  
  \begin{proof}
We want to show that $\varphi(n,K,q,s)$ goes to 0 for $n \to \infty$. 
First, we observe that the exponent of $q$ in  \eqref{varphi} is 
\begin{align*}
    n\ell-snK+sK^2 \geq Kn -Kns+sK^2 \geq K,
\end{align*}
 where for the first inequality we have used that $\ell \geq K$ and the second inequality follows from 
$   n> s(n-K)$, which in turn follows from $1/s \geq 1/2 >1-K/n.$
Thus, we get  the limit of $\varphi(n,K,q,s)$ for $K>n/2$:

\begin{align*}
& \lim\limits_{n \to \infty} \varphi(n,K,q,s) \\
  & =   \lim\limits_{n \to \infty} \left( (1/q)_K \sum\limits_{(k_1, \ldots, k_s) \in C(s,K)}q^{n\ell-snK+sK^2} \frac{(1/q)^{\sum_{i=1}^s \left(\sum_{j=1}^i k_j \right)^2}}{(1/q)_{k_1} \cdots (1/q)_{k_s}} \right)^{-1} \\
  & \leq \lim\limits_{n \to \infty} \frac{1}{(1/q)_K} \frac{1}{ q^{K}} \left( \sum\limits_{K_1, \ldots, K_s \geq 0} \frac{(1/q)^{\sum_{i=1}^s K_i^2}}{(1/q)_{k_1} \cdots (1/q)_{k_s}} \right)^{-1} \\  &  = 0.
\end{align*}

\end{proof}

 \begin{theorem}
  Let $\mathcal{R}$ have a residue field of size $q$ and let $s$ be its nilpotency index.  Let $K$ and $n$ be positive integers with $K= R'n$, where $0 < R' \leq 1/2$. 
 The density of   free submodules of $\mathcal{R}^n$ of given rank $K$ is  greater than or equal to  
$AGI(1/q,s)^{-1}.$ 
Moreover, if $R' < 1/2$, the density is 1. 
  \end{theorem}
  
  \begin{proof}

We wish to compute 
\begin{align*}  \lim\limits_{n \to \infty} \varphi(n,K,q,s)  & = \lim\limits_{n \to \infty} \left((1/q)_K \sum\limits_{(k_1, \ldots, k_s) \in C(s,K)}q^{n\ell-snK+sK^2} \frac{(1/q)^{\sum_{i=1}^s \left(\sum_{j=1}^i k_j \right)^2}}{(1/q)_{k_1} \cdots (1/q)_{k_s}} \right)^{-1} .
\end{align*}
Let us add and subtract to the exponent of $q$ a $\ell^2/s$, to get 
\begin{align*}
   \varphi(n,K,q,s) = \left( (1/q)_K \sum\limits_{(k_1, \ldots, k_s) \in C(s,K)}q^{n\ell-snK+sK^2-\ell^2/s} \frac{(1/q)^{\sum_{i=1}^s \left(\sum_{j=1}^i k_j \right)^2-\ell^2/s}}{(1/q)_{k_1} \cdots (1/q)_{k_s}} \right)^{-1}.
\end{align*}
We have that $n\ell-snK+sK^2-\ell^2/s \leq 0$ since $$K^2-(\ell/s)^2 \leq n(K-\ell/s),$$ which follows as $K+\ell/s \leq 2K \leq n.$ 
Hence, using a similar argument as in the proof of Theorem \ref{densbound} we get that 

\begin{align*}
  \lim\limits_{n \to \infty} \varphi(n,K,q,s) &= \lim\limits_{n \to \infty} \frac{\frac{1}{(1/q)_K } }{\sum\limits_{(k_1, \ldots, k_s) \in C(s,K)}q^{n\ell-snK+sK^2-\ell^2/s} \frac{(1/q)^{\sum_{i=1}^s \left(\sum_{j=1}^i k_j \right)^2-\ell^2/s}}{(1/q)_{k_1} \cdots (1/q)_{k_s}} } \\
  & \geq   \frac{1}{(1/q)_\infty}\left( \frac{1}{(1/q)_\infty} \sum\limits_{N_2, \ldots, N_s \geq 0} \frac{(1/q)^{\sum_{i=2}^s N_i^2}}{(1/q)_{n_2} \cdots (1/q)_{n_s}} \right)^{-1} \\ 
  & = AGI(1/q,s)^{-1}.
\end{align*}
Finally, we show that if $R'<1/2$, the density is 1. Let us consider again the probability $\varphi(n,K,q,s)$, which is 
$$\left( \sum_{(k_1, \ldots, k_s) \in C(s,K)} q^{n\ell-snK+sK^2 - \ell^2/s} (1/q)_K \frac{(1/q)^{\sum_{i=1}^s \left(\sum_{j=1}^i k_j \right)^2-\ell^2/s}}{(1/q)_{k_1} \cdots (1/q)_{k_s}} \right)^{-1}.$$
If we omit the type $(s^K (s-1)^0 \cdots 1^0)$ from the sum over $C(s,K)$ we get:  
\begin{align*}
 \left( 1 + \sum\limits_{\substack{(k_1, \ldots, k_s) \in C(s,K)  \\ k_1 \neq  K}}q^{n\ell-snK+sK^2-\ell^2/s} (1/q)_K \frac{(1/q)^{\sum_{i=1}^s \left(\sum_{j=1}^i k_j \right)^2 - \ell^2/s}}{(1/q)_{k_1} \cdots (1/q)_{k_s}}\right)^{-1}.
\end{align*}
Observe that whenever $\ell< s K$, the exponent of $q$ is  
$$n\ell-snK+sK^2-\ell^2/s \leq -(n-2K+1/s) \leq -(n-2K).$$ 
This follows from the fact, that $n\ell-snK+sK^2-\ell^2/s $ is decreasing for $\ell$ decreasing, and thus its largest value is reached at the largest $\ell \neq s K$, which is given by $K-1 + (s-1)/s$ corresponding to the type $(s^{K-1} (s-1)^1 (s-2)^ 0 \cdots 1^0).$

Now, we have that the probability of having a free module of given rank $K$ for $R'< 1/2$ is lower bounded by 
\begin{align*}
& \left( 1 + q^{-(n-2K)} \left( \sum\limits_{\substack{(k_1, \ldots, k_s) \in C(s,K)  \\ k_1 \neq  K}}(1/q)_K \frac{(1/q)^{\sum_{i=1}^s \left(\sum_{j=1}^i k_j \right)^2-\ell^2/s}}{(1/q)_{k_1} \cdots (1/q)_{k_s}}\right) \right)^{-1} \\ 
        & \geq \left( 1 + q^{-n(1-2R')} \left(\sum\limits_{K_2, \ldots, K_s \geq 0}  \frac{(1/q)^{K_2^2 + \cdots + K_s^2 -(K_2+ \cdots + K_s)^2/s}}{(1/q)_{k_2} \cdots (1/q)_{k_s}} \right) \right)^{-1},
\end{align*}
which clearly goes to 1, as $1-2R' >0$.
 
\end{proof}

\begin{table}[ht]
 \begin{center}
 \begin{tabular}{|c|c|c|c|c| }
 \hline 
 $q$ & $s$ & $K$ & $n$ & Probability  \\\hline
2 &2 & 50 & 100 & 0.460263  \\
2 & 2 & 40 & 100 &   0.999999 \\
2 & 2 & 60 & 100 &  $1.07 \cdot 10^{-31}$ \\
2 & 3 & 50 & 100  & 0.35536 \\
2 & 3 & 40 & 100  & 0.999999 \\
2 & 3 & 60 & 100  & $3.70 \cdot 10^{-62} $ \\
3 & 2 & 50 & 100 &   0.657496 \\ 
3 & 2 & 40 & 100 & $1-1.4\cdot 10^{-10}$   \\
3 & 2 & 60 & 100 &   $ 6.43\cdot 10^{-49}$  \\ \hline
  \end{tabular}
\end{center}  
  \caption{Probability of having a free submodule in $\mathcal{R}^n$ with rank $K$}
\end{table}

\section{Connections to Coding Theory}\label{sec:codes}

We turn now to the asymptotic behaviour of linear codes over a finite chain ring. This question is highly related to the results of the previous section, as a linear code over a finite chain ring $\mathcal{R}$ is simply an $\mathcal{R}$-module, albeit one that is studied in respect of a weight function defined on $\mathcal{R}^n$.

\begin{definition}
A (left) \emph{linear code} $\mathcal{C}$  over $\mathcal{R}$ of length $n$ is a (left) $\mathcal{R}$-submodule of $\mathcal{R}^n$. The free module $\mathcal{R}^n$ is called the {\em ambient space} of the code.
\end{definition}
 Since the length of a code already denotes the rank of the ambient space, we want to introduce a different parameter to denote the length $\ell$ of the module $\mathcal{C}.$ The usual parameter for ring-linear codes is also called type and either denotes size of the module, i.e., $\mid \mathcal{C} \mid =q^\ell$, or the frequencies of the parts of the type $\lambda$ of the module $\mathcal{C}$, i.e., $\left(\left(\pi^{s-1}\right)^{k_s} \cdots \left(1 \right)^{k_1} \right)$. However, to relate to the usual rate of a code,  we will call $$k:=\ell/s=\log_{|\mathcal{R}|}(|\mathcal{C} |) = \log_{q^s}(| \mathcal{C}|)$$ the $\mathcal{R}-$dimension of the code $\mathcal{C}$, since then the rate of $\mathcal{C}$ is given by $R= \frac{k}{n}$.
 

We summarize the density results derived in Section \ref{sec:dens}, which clearly apply immediately to linear codes.
Let $1\leq k \leq n$ be a fixed positive rational number with $ks \in \mathbb{N}$.
\begin{enumerate}
\item Let $0<\varepsilon <1$. 
For any $q$ satisfying $ (1/q)_\infty \geq 1-\varepsilon$, 
the probability that a random code of $\mathcal{R}$-dimension $k$ in $\mathcal{R}^n$ is free is at least $1- \varepsilon.$
    \item 
    The density of the free codes of $\mathcal{R}$-dimension $k$ in $\mathcal{R}^n$ is $1$ as $q \to \infty$.
        \item The density of the free codes of $\mathcal{R}$-dimension  $k$ in $\mathcal{R}^n$ is $(1/q)_\infty$ as $s \to \infty$.
    \item  Let $1/2 < R'< 1$ and let $K$ and $n$ be positive integers satisfying $K= R'n$. 
 The density of the free codes of rank $K$ in $\mathcal{R}^n$ is 0 as $n \to \infty$.
 \item  Let $0 < R' \leq 1/2$ and let $K$ and $n$ be positive integers satisfying $K= R'n$. 
 The density of the free codes of rank $K$ in $\mathcal{R}^n$ as $n \to \infty$ is at least $(1/q)_\infty.$ Moreover, if $R' < 1/2$, the density is 1. 
\end{enumerate}

The most prominent scalar ring arising in ring-linear coding is the Galois ring.
The $\mathbb{Z}/ 4 \mathbb{Z}$-linear codes are perhaps the most well-known of all codes over rings.
A great interest in codes over rings was generated by the discovery that several
infamous families of optimal but non-linear binary codes for the Hamming metric could be interpreted as linear codes over
$\mathbb{Z}/ 4 \mathbb{Z}$ for the Lee metric \cite{hammons}.
For codes over rings of odd characteristic, the Lee weight has been studied for its applications to partial response channels (see for example \cite{roth}). Codes defined over finite Frobenius rings are typically studied in relation to the homogeneous weight; due to its expression in terms of generating character of the underlying ring, several classical coding theoretic results extend for codes with respect to this metric (e.g. \cite{sneyd,grefsch,nech}). Note that a homogeneous weight is a weight function which takes the same average value on every non-zero one-sided ideal of the ring and hence every coordinate of a linear code carries the same total, respectively average, weight (see \cite{hom}).

\begin{example}
 The density of free codes of fixed rate in $\mathbb{Z}/4 \mathbb{Z}^n$, for $n\rightarrow \infty$, is given by $$\frac{2}{(-\sqrt{1/2}; 1/2)_\infty + (\sqrt{1/2}; 1/2)_\infty} \sim 0.59546.$$ 
\end{example}

The standard way to generate the codewords of a (left) linear code is via a {\em generator matrix}. 
For a (left) linear code $\mathcal{C}$ of rank $K$ defined over a finite ring, a generator matrix is defined to be any matrix $G$ satisfying $\mathcal{C} = \{ xG \mid x \in \mathcal{R}^K\}$.

Asymptotic properties of codes are often examined by creating a random code ensemble. One way to achieve this is via the generator matrix: a random linear code may be constructed by taking a random matrix and defining the code to be the row span of this matrix.

Hence, we study the probability that a random matrix of size $m \times n$ over $\mathcal{R}$ has $\mathcal{R}$-dimension  $k$. To count the number of $m \times n$ matrices over $\mathcal{R}$ of $\mathcal{R}$-dimension  $k$, we use the formulas derived in \cite{kschischang}.
   
\begin{proposition}[\text{\cite[Theorem 2]{kschischang}}]
The number of $m \times n$ matrices over $\mathcal{R}$  having type $(s^{k_1} \cdots  1^{k_s})$ is given by
$$T_{n,q}(k_1, \ldots, k_s)= q^{mks}\frac{(1/q)_m}{(1/q)_{m-K}} N_{n,q}(k_1, \ldots, k_s).$$
\end{proposition}

In the following we compute the density of  matrices of $\mathcal{R}$-dimension  $k$.
Let $n<m \in \mathbb{N}$. We call $A \in \mathcal{R}^{m \times n}$ \emph{rectangular unimodular} if there exist $m-n$ rows in $\mathcal{R}^n$, such that when adjoining these rows to $A$, the resulting $m \times m$ matrix is invertible over $\mathcal{R}$.
Thus, a rectangular unimodular matrix is the ring analog of a full rank matrix over a field.
\begin{proposition}\label{prop:densmat}
The density of  $m \times n$ matrices over $\mathcal{R}$ having $\mathcal{R}$-dimension  $k$, is $1$ if $k=m$ and 0 if $k<m.$
\end{proposition}
\begin{proof}
We note that it is enough to prove the result for the case $k=m$, because for $k<m$ the probability is bounded from above by the counter probability of the case $k=m$.  
Hence, we start with a random $k \times n$ matrix over $\mathcal{R}$ and ask for the module generated by this matrix to have $\mathcal{R}$-dimension  $k$.
Note that this is the same as asking for the matrix to be rectangular unimodular. The  probability that a $k \times n$ matrix over $\mathcal{R}$ is rectangular unimodular is 
$$ \frac{(1/q)_n}{(1/q)_{n-k}},$$ which does not depend on $s$ and goes to 1 for $n \to \infty$, also when $k=Rn$ for $0\leq R<1$. 

\end{proof}

\begin{remark}
We can give a lower bound on the probability in the case $k<m$, using  the fact that the number of all modules of $\mathcal{R}$-dimension  $k$ is greater than or equal to the number of free modules of $\mathcal{R}$-dimension  $k$, i.e., $$\gb{n}{ks}_q^{(s)} \geq q^{sk(n-k)} \gb{n}{k}_{1/q}.$$
Thus, we get that \begin{align*}
& \frac{\sum\limits_{(k_1, \ldots, k_s) \in L(s,n,ks)}q^{smk} \frac{(1/q)_m}{(1/q)_{m-K}} N_{n,q}(k_1, \ldots, k_s)}{q^{smn}} \\
\geq &    (1/q)_m (1/q)^{sm(n-k)} \gb{n}{ks}_q^{(s)} \\
\geq  & \frac{(1/q)_m (1/q)_n}{(1/q)_{n-k} (1/q)_k} (1/q)^{s(n-k)(m-k)}.
\end{align*}
\end{remark}

\subsection{Random Codes Achieve the Gilbert-Varshamov Bound}

In algebraic coding theory one of the most important parameters of a code is its minimum distance. For this we have to endow our finite chain ring with a metric. 

\begin{definition}\label{weight}
We say that a map 
    $\text{wt}: \mathcal{R} \to \mathbb{R}$ is a {\em weight} if it satisfies the following:
\begin{enumerate}
    \item $\text{wt}$ is positive definite, i.e., $\text{wt}(0) =0$, and $\text{wt}(x) >0 $ for all $x \in \mathcal{R}\setminus\{0\},$
    \item $\text{wt}$ is symmetric, i.e., $\text{wt}(x)=\text{wt}(-x)$ for all $x \in \mathcal{R}$,
    \item $\text{wt}$ satisfies the triangle inequality, i.e, $\text{wt}(x+y) \leq \text{wt}(x) + \text{wt}(y),$ for all $x,y \in \mathcal{R}.$
\end{enumerate}
\end{definition}
One can extend this weight additively to a weight on $\mathcal{R}^n$. 
By abuse of notation, the weight on $\mathcal{R}^n$ will also be denoted by $\text{wt}.$
Then for any $x =(x_1,\ldots,x_n)\in \mathcal{R}^n$ we have 
$\text{wt}(x) := \sum_{i=1}^n \text{wt}(x_i)$.
Such a weight induces a distance function 
\begin{align*}
    d: \mathcal{R} \times \mathcal{R} &\to  \mathbb{R}, \\
    (x,y) & \mapsto \text{wt}(x-y).
\end{align*} Again, we may extend this additively to a distance function $d:\mathcal{R}^n \times \mathcal{R}^n \longrightarrow {\mathbb R}$.

The minimum distance of a code $\mathcal{C} \in \mathcal{R}^n$ is defined to be 
$$d(\mathcal{C}) := \min\{d(x,y) \mid x, y \in \mathcal{C}, x \neq y\}.$$

We will furthermore assume that any distance function considered in this section is translation-invariant, that is, $d(x,y)=d(x+a,y+a)$ for all $a,x,y \in \mathcal{R}^n.$
One can endow $\mathcal{R}$ with the following weights, which all induce translation invariant distance functions on $\mathcal{R}^n$.
\begin{definition}\label{def:exwt}

\begin{itemize}
    \item[i)] The \emph{Hamming weight} of $x \in \mathcal{R}^n$ is given by the number of non-zero entries, i.e.,
$$\text{wt}_H(x) := \mid \{ i \in \{1, \ldots, n \} \mid x_i \neq 0 \} \mid .$$
\item[ii)] The \emph{homogeneous weight} of  $x \in \mathcal{R}$ is defined to be \[\text{wt}_{\hom}(x) :=\left\{\begin{array}{cl}\frac{q}{q-1} & \text{ if }x \in \langle \pi^{s-1} \rangle \backslash \{0\}, \\
               1 &\text{ if } x \notin \langle \pi^{s-1} \rangle, \\
               0 & \text{ otherwise.} 
               \end{array}\right.\] 
\item[iii)] The \emph{Lee weight} of $x \in \zps$ is given by 
$$\text{wt}_L(x) := \min \{x,  \mid p^s -x \mid  \}.$$
\end{itemize}
\end{definition}

Note, in the case that $\mathcal{R}=\F_q$, for the formulation of the homogeneous weight given here, the Hamming weight represents a normalization of the homogeneous weight, while for $\mathcal{R}=\Z/4\Z$, the Lee weight and the homogeneous weight coincide.

One of the most important bounds in coding theory is the Gilbert-Varshamov (G-V) bound, which gives a bound on the maximal possible size a code can have given its length $n$ and its minimum distance. In its asymptotic version, the G-V bound gives a lower bound on the rate $R=k/n$ of a code. 
For codes over finite fields it is well known that random codes achieve the G-V bound with high probability \cite{bargforney,pierce}. The same result for rank-metric codes was also shown using the probabilistic method in \cite{loidreau}. 
While the G-V bound for codes over finite chain rings has been addressed in \cite{ozbsole}, the statement on random codes is missing in the literature. 

In the rest of this section, we prove that a random code over a finite chain ring $\mathcal R$ endowed with a metric achieves, with high probability, the asymptotic Gilbert-Varshamov bound.

 Let $N$ be the maximal weight that an element of $\mathcal{R}^n$ can achieve.  For example in the Hamming metric we have $N=n$, for the homogeneous weight we have $N= \frac{q}{q-1}n$, while in the Lee metric on $\zpsk{n}$ we have $N= n \lfloor \frac{p^s}{2} \rfloor$.

The open ball of radius $w$ centred at $u \in \mathcal{R}^n$ is defined by 
\begin{gather*} 
B(n,x,w) := \{y \in \mathcal{R}^n \mid d(x,y)< w \}.
\end{gather*}
Since $d$ is translation-invariant, we have that 
$|B(n,x,w)|$ is independent of its center and in particular is equal to $|B(n,0,w)|$.
We denote its size by $V(n,w):=|B(n,0,w)|$.
We write $\bar{B}(n,x,w)$ to denoted the closed ball centred at $x$, that is
$\bar{B}(n,x,w) := \{y \in \mathcal{R}^n \mid d(x,y)\leq w \}$, which is again translation-invariant and we write 
$\bar{V}(n,w)$ to denote its cardinality, $|\bar{B}(n,0,w)|$.  
Let us denote by $AL(n,d)$ the maximal size of a code in $\mathcal{R}^n$ having minimum distance $d$.

The Gilbert-Varshamov bound  \cite{gilbert, varshamov} now states that
\begin{equation*}
  AL(n,d) \geq \frac{q^{sn}}{V(n,d)}.  
\end{equation*}
 
We denote by 
$$g(\delta) := \lim\limits_{n \to \infty} \frac{1}{n} \log_{q^s}\left( \bar{V}(n, \delta N) \right). $$ 
In \cite{loeliger} it is shown that this limit exists for an arbitrary weight function defined additively on its coordinates and that it is equal to the entropy of the corresponding probability distribution on $\mathcal R$.
For example, in the Hamming metric it is well known \cite{loeliger} that 
$$g(\delta) = h_{q^s}(\delta),$$ where $h_{q^s}$ is the $q^s$-ary entropy function.
In the case of homogeneous weight, the limit was explicitly computed in \cite{Greferath}, and for the Lee metric on $\zpsk{n}$, it  was computed in \cite{saddle, leenp}.

Let us denote by $$\overline{R}(\delta) := \limsup\limits_{n \to \infty} \frac{1}{n} \log_{q^s} AL(n, \delta N). $$

The asymptotic Gilbert-Varshamov bound now states that 
$$ \overline{R}(\delta) \geq 1-g(\delta). $$

Note that $g(\delta) \in [0,1]$ is a continuous and increasing function in $\delta$. Let us denote by $D$ the minimum value in $[0,1]$ such that $g(D) = 1.$  For example in the Hamming metric we have  $D=1-1/q^s$, in the homogeneous metric we have $D=1$, and in the Lee metric over $\zps$ we have $D\approx 1/2.$

Recall that in complexity theory we write $f(n) = \Omega(g(n))$, if $\limsup\limits_{n \to \infty} \left | \frac{f(n)}{g(n)} \right | >0$.
For example, $f(n) = \Omega(n)$ means that $f(n)$ grows at least polynomially in $n$.

\begin{theorem}
Let $\mathcal R$ be a finite chain ring having maximal ideal $\langle \pi \rangle$, residue field of size $q$ and nilpotency index $s$. Then for every $\delta \in [0, D),$ $$0 < \varepsilon < 1- g(\delta N/n),$$ and sufficiently large $n$, the following holds for $$k= \lceil(1- g(\delta N/n) - \varepsilon)n \rceil.$$ 
If $G \in  \mathcal R^{k \times n}$ is chosen uniformly at random, then the code generated by $G$ has rate $k/n$ and relative minimum distance $\delta$ with probability at least $$\left(\frac{(1/q)_n}{(1/q)_{n-k}} \right) \left(1- q^{s\left(1-\varepsilon n\right)}\right) \geq 1-  e^{-\Omega(n)}.$$ 
\end{theorem}

\begin{proof}
The probability that a matrix $G \in \mathcal R^{k \times n}$ generates a code of $\mathcal{R}$-dimension $k$  is given by $$\frac{(1/q)_n}{(1/q)_{n-k}}$$ as seen in the proof of Proposition \ref{prop:densmat}.

For any non-zero $x \in \mathcal R^{k}$, we will show that the probability that $\wt(xG) \leq \delta N/n$ is at most 
\[\frac{  \bar{V}(n, \delta N/n) }{q^{sn}} \leq  q^{s\left(g(\delta N/n)n-n\right)}.\]
This clearly follows if we assume $x$ is unimodular, since then 
the morphism $f_x: \mathcal{R}^k \longrightarrow \mathcal{R}, z \mapsto x \cdot z$ is a surjection and so $xG$ is a uniformly distributed in $\mathcal{R}^n$ for $G$ chosen uniformly at random in $\mathcal{R}^{k \times n}$.

More generally, let $x \in \pi^i \mathcal{R}^k \setminus \pi^{i+1}\mathcal{R}^k$ for some $i \in \{1,\ldots,s-1\}$, i.e., $x = \pi^i y$ for some unimodular $y \in \mathcal R^{k}$. 
For $$\eta = \frac{\max\{\wt(z) \mid z \in \mathcal{R}^n\}}{\min\{\wt(z) \mid z \in \mathcal{R}^n \setminus \{0\}\}} \geq 1,$$ we see that $\wt(xG) \leq \eta \wt(yG)$. Moreover, $yG$ is uniformly random in $\mathcal R^n$.\\
Finally, the probability that $\wt(xG) \leq \delta N/n$ is less than or equal to the probability that $\wt(yG) \leq (\delta N/n)/\eta$, which is upper bounded by
\[\frac{  \bar{V}(n, (\delta N/n)/\eta) }{q^{sn}} \leq \frac{  \bar{V}(n, \delta N/n) }{q^{sn}} \leq q^{s\left(g(\delta N/n)n-n\right)}.\]

Using a union bound over all non-zero $x$ implies that the code has minimum distance at most $\delta n$ is bounded from above by 
$$q^{sk} q^{s\left(g(\delta N/n)n-n\right)}  \leq q^{s(1-\varepsilon n)}.$$
Hence, we have that the code $\mathcal{C}$ has relative minimum distance $\delta$ with probability at least $   \left(1- q^{s\left(1-\varepsilon n\right)}\right)\geq 1- e^{-\Omega(n)}$. 

\end{proof}

 \section{Conclusion}\label{sec:concl}
We conclude this paper with two interesting open problems we encountered. The first problem is to find another generalization of the Roger-Ramanujan identities. Since the density of free modules is bounded by the Andrews-Gordon identity (see Theorem \ref{thm:dens}), which themselves are products of $q$-Pochhammer symbols, this leads us to the question of finding a similar identity for the density of free submodules of $\mathcal{R}^n$ of given length. More precisely, we state the following problem.
\begin{problem}\label{pr1}
Determine whether \begin{equation*}
    \left(\sum\limits_{\substack{K_2, \ldots, K_s \geq 0\\ s \mid K_2 + \cdots + K_s }} \frac{(1/q)^{K_2^2 + \cdots+ K_s^2  -(K_2+ \cdots +K_s)^2/s  }}{(1/q)_{k_2} \cdots (1/q)_{k_s}} \right)^{-1},
\end{equation*}
where $K_i=\sum_{j=2}^i k_j$, has an expression as product of $q$-Pochhammer symbols. 
\end{problem}

The second problem we have encountered is that there does not seem to exist a natural way to order the types in $L(s,n,\ell)$ such that the canonical order on the number of modules $N_{n,q}(k_1, \ldots, k_s)$ is induced by this order on $L(s,n,\ell).$  Such a characterization would be convenient in order to determine which type is more likely. For example, for fixed finite parameters, does the number of free submodules dominate the number of modules of any other type?

\begin{example}
\begin{enumerate}
\item Note, that the lexicographic order on $L(s,n,\ell)$ does not satisfy this property, for example let $n=10, \ell=15, s=3$ and $q=2$, then $(4,0,3) >_{lex} (3,3,0)$ but there are more codes of the type $(3^3 2^3 1^0)$ than of the type $(3^4 2^0 1^3)$.
\item Also ordering  them according to their rank $K$, i.e., $(k_1, \ldots, k_s) >_K (\bar{k}_1, \ldots, \bar{k}_s)$ if $\sum_{i=1}^s \bar{k}_i > \sum_{i=1}^s k_i$, does not work, since for $n=10, \ell=15, s=3$ and $q=2$,  the type $(3^2 2^3 1^3)$ has a larger rank than $(3^1 2^6 1^0)$ but there are more codes of type $(3^1 2^6 1^0)$ than of the type $(3^2 2^3 1^3).$
\item Finally, let us define the following order (which is related to the shape $\mu$): $(k_1, \ldots, k_s) >_s (\bar{k}_1, \ldots, \bar{k}_s)$ whenever $$\sum_{i=1}^s \left( \sum_{j=1}^i \bar{k}_j \right)^2 > \sum_{i=1}^s \left( \sum_{j=1}^i k_j \right)^2.$$ Also this order does not solve the problem: for $n=10,\ell=30, s=6$ and $q=2$ we have that $(2,1,1,1,2,2) >_s (0,5,1,0,0,1)$ but $N_{n,q}(0,5,1,0,0,1) > N_{n,q}(2,1,1,1,2,2).$
\end{enumerate}
\end{example}

\begin{problem}\label{proborder}
Establish a simplified condition on $(k_1, \ldots, k_s),  (\bar{k}_1, \ldots, \bar{k}_s) \in L(s,n,\ell)$ such that we have $$N_{n,q}(k_1, \ldots, k_s) \leq N_{n,q}(\bar{k}_1, \ldots, \bar{k}_s).$$
\end{problem}

\begin{remark}
   We remark that since this paper was submitted, Problem \ref{pr1} was addressed in \cite{do}, wherein the authors showed that the quantity in Problem \ref{pr1} can be expressed as a sum of $s^2$ quotients of theta functions, rather than as a product of $q$-Pochhammer symbols.
\end{remark}

\section*{Acknowledgments}\label{sec:ack}
The authors want to thank Gianira N. Alfarano and the anonymous referee for fruitful discussions on the shape of a module. The third author is supported by the Estonian Research Council grant PRG49. The fourth author  is  supported by the Swiss National Science Foundation grant number 195290.

\bibliographystyle{plain}

\bibliography{References}

\begin{thebibliography}{10}

\bibitem{alder}
Henry~L. Alder.
\newblock {Generalizations of the Rogers-Ramanujan identities.}
\newblock {\em Pacific Journal of Mathematics}, 4(2):161 -- 168, 1954.

\bibitem{andrewid}
George~E. Andrews.
\newblock An analytic generalization of the {R}ogers-{R}amanujan identities for
  odd moduli.
\newblock {\em Proceedings of the National Academy of Sciences},
  71(10):4082--4085, 1974.

\bibitem{baxter}
George~E. Andrews and Rodney~J. Baxter.
\newblock Lattice gas generalization of the hard hexagon model. iii.
  {$q$}-trinomial coefficients.
\newblock {\em Journal of statistical physics}, 47(3):297--330, 1987.

\bibitem{heide}
Jared Antrobus and Heide Gluesing-Luerssen.
\newblock Maximal {F}errers diagram codes: Constructions and genericity
  considerations.
\newblock {\em IEEE Transactions on Information Theory}, 65(10):6204--6223,
  2019.

\bibitem{assmus}
EF~Assmus~Jr and Harold~F. Mattson.
\newblock Error-correcting codes: An axiomatic approach.
\newblock {\em Information and Control}, 6(4):315--330, 1963.

\bibitem{bailey}
Wilfrid~Norman Bailey.
\newblock {\em Generalized hypergeometric series}.
\newblock Stechert-Hafner, 1964.

\bibitem{Balakin}
Gennadii~Vasil'evich Balakin.
\newblock The distribution of the rank of random matrices over a finite field.
\newblock {\em Theory of Probability and Its Applications}, 13(4):594–605,
  1968.

\bibitem{bargforney}
Alexander Barg and G.~David Forney.
\newblock Random codes: Minimum distances and error exponents.
\newblock {\em IEEE Transactions on Information Theory}, 48(9):2568--2573,
  2002.

\bibitem{bariffi2021analysis}
Jessica Bariffi, Hannes Bartz, Gianluigi Liva, and Joachim Rosenthal.
\newblock Analysis of low-density parity-check codes over finite integer rings
  for the {L}ee channel.
\newblock {\em https://arxiv.org/abs/2105.08372}, 2021.

\bibitem{blake1}
Ian~F. Blake.
\newblock Codes over certain rings.
\newblock {\em Information and Control}, 20(4):396--404, 1972.

\bibitem{blake2}
Ian~F. Blake.
\newblock Codes over integer residue rings.
\newblock {\em Information and Control}, 29(4):295--300, 1975.

\bibitem{butler}
Lynne~M. Butler.
\newblock {\em Subgroup lattices and symmetric functions}, volume 539.
\newblock American Mathematical Soc., 1994.

\bibitem{partition}
Eimear Byrne and Alberto Ravagnani.
\newblock Partition-balanced families of codes and asymptotic enumeration in
  coding theory.
\newblock {\em J. Comb. Theory, Ser. A}, 171, 2020.

\bibitem{sneyd}
Eimear Byrne and Alison Sneyd.
\newblock Two-weight codes, graphs and orthogonal arrays.
\newblock {\em Designs, Codes and Cryptography}, 79(2):201--217, 2016.

\bibitem{chara}
Charalambos~A. Charalambides.
\newblock {\em Discrete $q$-distributions}.
\newblock John Wiley \& Sons, 2016.

\bibitem{hom}
Ioana Constantinescu.
\newblock {\em Lineare Codes {\"u}ber {R}estklassenringen ganzer {Z}ahlen und
  ihre {A}utomorphismen bez{\"u}glich einer verallgemeinerten
  {H}amming-{M}etrik}.
\newblock PhD thesis, Technische Universit\"at M\"unchen, 1995.

\bibitem{count}
Steven~T. Dougherty and Eseng{\"u}l Salt{\"u}rk.
\newblock Counting codes over rings.
\newblock {\em Designs, codes and cryptography}, 73(1):151--165, 2014.

\bibitem{do}
Jehanne Dousse and Robert Osburn.
\newblock A $q$-multisum identity arising from finite chain ring probabilities.
\newblock {\em arXiv preprint arXiv:2111.11123v1}, 2021.

\bibitem{kschischang}
Chen Feng, Roberto~W. N{\'o}brega, Frank~R. Kschischang, and Danilo Silva.
\newblock Communication over finite-chain-ring matrix channels.
\newblock {\em IEEE Transactions on Information Theory}, 60(10):5899--5917,
  2014.

\bibitem{fulmangoldstein}
Jason Fulman and Larry Goldstein.
\newblock {Stein’s method and the rank distribution of random matrices over
  finite fields}.
\newblock {\em The Annals of Probability}, 43(3):1274 -- 1314, 2015.

\bibitem{gadouleau}
Maximilien Gadouleau.
\newblock {\em Algebraic codes for random linear network coding}.
\newblock Lehigh University, 2009.

\bibitem{gallager}
Robert Gallager.
\newblock The random coding bound is tight for the average code (corresp.).
\newblock {\em IEEE Transactions on Information Theory}, 19(2):244--246, 1973.

\bibitem{saddle}
Dani{\`e}le Gardy and Patrick Sol{\'e}.
\newblock Saddle point techniques in asymptotic coding theory.
\newblock In {\em Workshop on Algebraic Coding}, pages 75--81. Springer, 1991.

\bibitem{gasper}
George Gasper, Mizan Rahman, and Gasper George.
\newblock {\em Basic hypergeometric series}, volume~96.
\newblock Cambridge university press, 2004.

\bibitem{gilbert}
Edgar~N. Gilbert.
\newblock A comparison of signalling alphabets.
\newblock {\em The Bell system technical journal}, 31(3):504--522, 1952.

\bibitem{gordon}
Basil Gordon.
\newblock A combinatorial generalization of the {R}ogers-{R}amanujan
  identities.
\newblock {\em American Journal of Mathematics}, 83(2):393--399, 1961.

\bibitem{nech}
Marcus Greferath, Alexandr Nechaev, and Robert Wisbauer.
\newblock Finite quasi-{F}robenius modules and linear codes.
\newblock {\em Journal of Algebra and Its Applications}, 03, 11 2011.

\bibitem{Greferath}
Marcus Greferath and Michael~E. O’Sullivan.
\newblock On bounds for codes over {F}robenius rings under homogeneous weights.
\newblock {\em Discrete Mathematics}, 289(1):11--24, 2004.

\bibitem{grefsch}
Marcus Greferath and Stefan~E. Schmidt.
\newblock Finite-ring combinatorics and {MacW}illiams equivalence theorem.
\newblock {\em J. of Combinatorial Theory (A)}, 92, 2000.

\bibitem{gruicarav}
Anina Gruica and Alberto Ravagnani.
\newblock Common complements of linear subspaces and the sparseness of {MRD}
  codes.
\newblock {\em https://arxiv.org/abs/2011.02993}, 2021.

\bibitem{hammons}
A.~Roger Hammons, P.~Vijay Kumar, A.~Robert Calderbank, Neil~J.A. Sloane, and
  Patrick Sol{\'e}.
\newblock The {$\mathbb{Z}_4$}-linearity of {K}erdock, {P}reparata, {G}oethals,
  and related codes.
\newblock {\em IEEE Transactions on Information Theory}, 40(2):301--319, 1994.

\bibitem{numbermod}
Thomas Honold and Ivan Landjev.
\newblock Linear codes over finite chain rings.
\newblock {\em the electronic journal of combinatorics}, 7:R11--R11, 2000.

\bibitem{Horlemann2019}
Anna-Lena Horlemann-Trautmann and Violetta Weger.
\newblock Information set decoding in the {L}ee metric with applications to
  cryptography.
\newblock {\em Advances in Mathematics of Communications}, online, 2019.

\bibitem{klemm}
Michael Klemm.
\newblock {\"U}ber die {I}dentit{\"a}t von {MacW}illiams f{\"u}r die
  {G}ewichtsfunktion von {C}odes.
\newblock {\em Archiv der Mathematik}, 49(5):400--406, 1987.

\bibitem{loeliger}
Hans-Andrea Loeliger.
\newblock An upper bound on the volume of discrete spheres.
\newblock {\em IEEE Transactions on Information Theory}, 40(6):2071--2073,
  1994.

\bibitem{loidreau}
Pierre Loidreau.
\newblock Asymptotic behaviour of codes in rank metric over finite fields.
\newblock {\em Des. Codes Cryptogr.}, 71:105--118, 2014.

\bibitem{mac}
Ian~Grant Macdonald.
\newblock {\em Symmetric functions and {H}all polynomials}.
\newblock Oxford university press, 1998.

\bibitem{nechaev}
Aleksandr~Aleksandrovich Nechaev.
\newblock Kerdock's code in cyclic form.
\newblock {\em Diskretnaya matematika}, 1(4):123--139, 1989.

\bibitem{niven}
Ivan Niven.
\newblock {The asymptotic density of sequences}.
\newblock {\em Bulletin of the American Mathematical Society}, 57(6):420 --
  434, 1951.

\bibitem{ozbsole}
Ferruh \"{O}zbudak and Patrick Sol\'e.
\newblock Gilbert-{V}arshamov type bounds for linear codes over finite chain
  rings.
\newblock {\em Advances in the Mathematics of Communications}, 1(1):99--109,
  2007.

\bibitem{pierce}
John~N. Pierce.
\newblock Limit distribution of the minimum distance of random linear codes.
\newblock {\em IEEE Transactions on Information Theory}, 13(4):595--599, 1967.

\bibitem{roth}
Ron~M. Roth and Paul~H. Siegel.
\newblock Lee-metric {BCH} codes and their application to constrained and
  partial-response channels.
\newblock {\em IEEE Transactions on Information Theory}, 40(4):1083--1096,
  1994.

\bibitem{gengauss}
Eseng{\"u}l Salt{\"u}rk and {\.I}rfan {\c{S}}iap.
\newblock Generalized {G}aussian numbers related to linear codes over {G}alois
  rings.
\newblock {\em European Journal of Pure and Applied Mathematics},
  5(2):250--259, 2012.

\bibitem{schilling1996multinomials}
Anne Schilling.
\newblock Multinomials and polynomial bosonic forms for the branching functions
  of the {$\widehat{su}_M(2)\times \widehat{su}_N(2) / \widehat{su}_{M+N}(2)$}
  conformal coset models.
\newblock {\em Nuclear Physics B}, 467(1-2):247--271, 1996.

\bibitem{supernom}
Anne Schilling and S.~Ole Warnaar.
\newblock Supernomial coefficients, polynomial identities and {$q$}-series.
\newblock {\em The Ramanujan Journal}, 2(4):459--494, 1998.

\bibitem{shannon}
Claude~E. Shannon.
\newblock A mathematical theory of communication.
\newblock {\em Bell System Technical Journal}, 27(3):379--423, 1948.

\bibitem{shiromoto}
Keisuke Shiromoto.
\newblock Singleton bounds for codes over finite rings.
\newblock {\em Journal of Algebraic Combinatorics}, 12(1):95--99, 2000.

\bibitem{spiegel}
Eugene Spiegel.
\newblock Codes over {$\mathbb{Z}_m$}.
\newblock {\em Information and control}, 35(1):48--51, 1977.

\bibitem{varshamov}
Rom~Rubenovich Varshamov.
\newblock Estimate of the number of signals in error correcting codes.
\newblock {\em Docklady Akad. Nauk, SSSR}, 117:739--741, 1957.

\bibitem{qmultinomials}
S.~Ole Warnaar.
\newblock The {A}ndrews--{G}ordon identities and {$q$}-multinomial
  coefficients.
\newblock {\em Communications in mathematical physics}, 184(1):203--232, 1997.

\bibitem{leenp}
Violetta Weger, Karan Khathuria, Anna-Lena Horlemann-Trautmann, Massimo
  Battaglioni, Paolo Santini, and Edoardo Persichetti.
\newblock On the hardness of the {L}ee syndrome decoding problem.
\newblock {\em arXiv preprint arXiv:2002.12785}, 2020.

\end{thebibliography}

\end{document}